\documentclass[runningheads]{llncs}

\newcommand{\reachSymbol}{\ensuremath{\mathcal{R}}}
\newcommand{\reach}[2]{\ensuremath{\reachSymbol_{[#1, #2]}}\xspace}
\newcommand{\oreach}[2]{\ensuremath{\widehat{\reachSymbol}_{[#1, #2]}}\xspace}
\newcommand{\X}{\ensuremath{\mathcal{X}}\xspace}
\newcommand{\U}{\ensuremath{\mathcal{U}}\xspace}
\newcommand{\safe}{\ensuremath{\X_s}\xspace}
\newcommand{\R}{\ensuremath{\mathbb{R}}\xspace}
\newcommand{\CH}{\ensuremath{\overline{\mathit{CH}}}\xspace}
\newcommand{\CHexact}{\ensuremath{\mathit{CH}}\xspace}
\newcommand{\policy}{\ensuremath{\pi}\xspace}

\newcommand{\TheTitle}{Algorithm~\ref{alg:ReachTube}\xspace}

\newcommand{\figref}[1]{Fig.~\ref{#1}\xspace}

\usepackage[utf8]{inputenc}
\usepackage[main=english]{babel}
\usepackage[T1]{fontenc}
\usepackage[babel,threshold=2]{csquotes}

\usepackage[labelformat=simple]{subcaption}

\usepackage{amssymb}
\usepackage{amsmath}

\PassOptionsToPackage{hyphens}{url} 

\usepackage{hyperref}
\usepackage{color}

\usepackage{url}

\usepackage{algorithm}
\usepackage{algpseudocode}

\usepackage{graphicx}

\usepackage[dvipsnames,table,xcdraw]{xcolor}

\usepackage{setspace}

\usepackage{longtable}
\usepackage{booktabs} 
\usepackage{siunitx}  
\usepackage{makecell} 

\usepackage{xspace}

\RequirePackage{graphicx}
\usepackage[firstpageonly=true,angle=0,vpos=.147\paperheight,hpos=.39\linewidth,vanchor=t,hanchor=l]{draftwatermark}
\SetWatermarkText{%
{\begin{minipage}{\linewidth}\centering\includegraphics[width=12.5mm]
{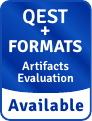}\hfill\includegraphics[width=12.5mm]
{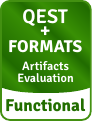}\end{minipage}}%
}

\begin{document}

\title{Scalable Reachability Analysis of \\ Linear Continuous Systems with \\ Property-Driven Time-Step Adaptation}
\titlerunning{Reachability Analysis with Property-Driven Time-Step Adaptation}
%
\author{Mikkel~Bjørn\inst{1}\orcidID{0009-0004-7888-9070} \and
Daniel~H.~Hansen\inst{1}\orcidID{0009-0006-3085-7293} \and
Grace~Melchiors\inst{1} \and
Kim~G.~Larsen\inst{1}\orcidID{0000-0002-5953-3384} \and
Christian~Schilling\inst{1}\orcidID{0000-0003-3658-1065}}
\authorrunning{M. Bjørn \and D. H. Hansen \and G. Melchiors \and K. G. Larsen \and C. Schilling}
%
\institute{Aalborg University, Aalborg, Denmark}

\maketitle

\begin{abstract}
We study safety verification for linear time-invariant systems with bounded inputs in continuous time. The standard approach reduces to a reachability analysis in two steps: first discretize time and then apply a forward analysis in the discretized system. Existing algorithms use either a fixed time step or an adaptive time step that changes based on the approximation error compared to the underlying continuous system. In this paper, we present an efficient reachability algorithm that adapts the time step based on a given safety property. Essentially, our algorithm makes the largest possible time step such that it can still prove safety. For this approach to be scalable in practice, we discuss several optimizations such as avoiding the repeated expensive calculation of the matrix exponential during discretization and a careful balance how we tame the approximation error stemming from the states and the inputs. This allows our algorithm to yield a moderate approximation error even when using a large time step, thus requiring much fewer steps than prior algorithms. We demonstrate the effectiveness and scalability on the large-scale SLICOT benchmark suite, where our algorithm consistently outperforms other state-of-the-art approaches.

\end{abstract}

\section{Introduction}
Verification of safety properties for linear time-invariant (LTI) systems is a well-studied problem~\cite{BasicNotation,LGG09,Combiningzonotopesandsupportfunctions,DecomposingReachSetComputations,Undecidable:1,Luo2023,Wetzlinger2026,ReachabilityOverview}.\
In this paper, we focus on the common setting with continuous-time dynamics described by ordinary differential equations (ODEs), where the initial states and inputs are uncertain within a bounded domain, and where time is bounded.
Formally, our systems have the form
\begin{equation}\label{eq:LTI}
    \dot{x}(t) = Ax(t)+u(t), \quad t\in \left[0, T\right],
\end{equation}
where the initial state~$x(0) \in \X_0$ is constrained to a set of initial states, $u(t)\in \U$ is the input signal constrained to an input domain, and $T>0$ is the time horizon.

Given a safety property in the form of a set of safe states~$\safe$, the safety verification problem is to determine whether all trajectories emerging from~$\X_0$ with input signal~$u(t) \in \U$ stay safe within time~$T$.
The problem thus reduces to computing the set of reachable states~$\reach{0}{T}(\X_0)$ up to time horizon $T$ and checking that~$\reach{0}{T}(\X_0)$ is included in~$\safe$.
The set~$\reach{0}{T}(\X_0)$ is generally not computable except for simple subclasses~\cite{LafferrierePY01,Undecidable:1,ChenG24}.
Hence, one typically computes an overapproximation~$\oreach{0}{T}(\X_0) \supseteq \reach{0}{T}(\X_0)$; this is sufficient, as $\oreach{0}{T}(\X_0) \subseteq \safe$ implies $\reach{0}{T}(\X_0) \subseteq \safe$.
The standard approach to obtaining~$\oreach{0}{T}(\X_0)$ is to (1)~discretize the continuous system and (2)~compute a series of discrete steps in the resulting discrete system~\cite{ForetsS22}.
The approximation quality can be increased by reducing the time step of the discretization, which, however, then necessitates a larger number of steps to be taken.

While most approaches use a fixed time step and thus require a single discretization operation, some more advanced approaches instead use an adaptive time step~\cite{PrabhakarV11,FrehseKG13,KOCHDUMPER2024101491,ErrorBoundBased,SpaceEx:errorsupportFunc}.
The latter approaches have in common that they analyze the approximation error (based on error estimates) dynamically and choose a time step such that the error is below a user-defined threshold at all times.

\smallskip

\begin{figure}[t]
    \centering
    \includegraphics[width=\textwidth,height=65mm,keepaspectratio]{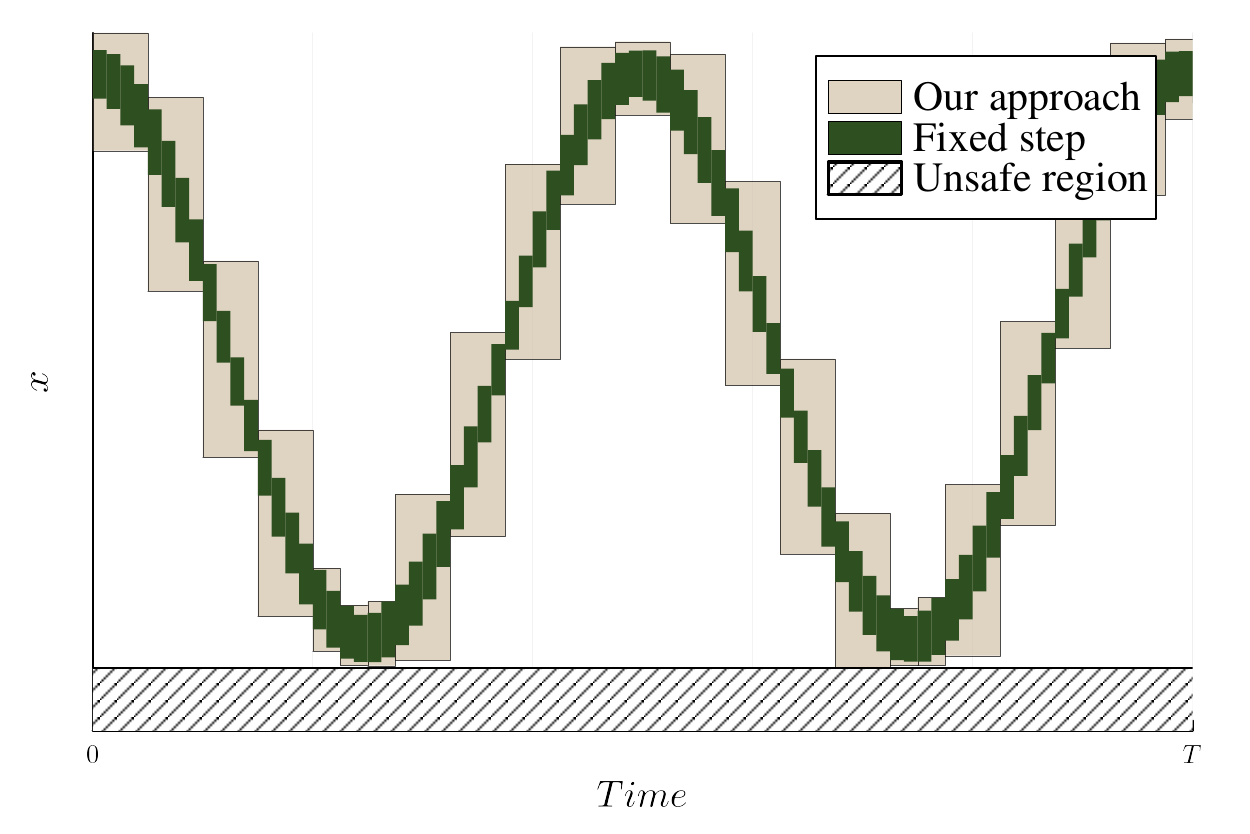}
    \caption{Two overapproximations~$\oreach{0}{T}(\X_0)$ of the reachable states.
    Each block corresponds to one step in time.
    The green sequence was obtained by a fixed-step approach with a step size chosen such that safety (staying above the bottom unsafe region) is verified.
    The brown sequence was obtained by our approach, which chooses the time step depending on the distance to the unsafe region.}
    \label{fig:motivation}
\end{figure}

In this paper, we propose a new approach that uses neither a fixed time step nor a fixed approximation error.
As a motivation, we now discuss the two main drawbacks of these two existing paradigms.
First, both paradigms scale poorly when high precision is required \emph{at some point} during the analysis.
This is clear for approaches using a fixed time step, as the time step must be chosen small enough for such critical points, even if a larger step would be sufficient most of the time.
Approaches using a fixed error can sometimes make larger steps, but they are still held back by the prescribed error bound, and, just like above, this error bound must be small enough; hence, the variance in step sizes is often rather modest in practice.
This drawback is illustrated in \figref{fig:motivation}, where the adaptive algorithm would have to use a similar error as the fixed-step algorithm, since the approximation~$\oreach{0}{T}(\X_0)$ gets very close to the error states.

The second drawback is that the user is expected to anticipate the required precision up-front.
Either a small enough step size or an allowed approximation error must be chosen, but it is initially unknown how closely the reachable states (and hence the approximation~$\oreach{0}{T}(\X_0)$) will approach the boundaries of~$\safe$.

Our approach does not have these two drawbacks, i.e., it can make large time steps for better scalability and does not require an error bound.
Specifically, we assume that the user only wants to verify safety as fast as possible and does not care about the approximation error.
At a high level, our approach is optimistic and tries to make large time steps.
If at some point the precision is insufficient to prove $\oreach{0}{T}(\X_0) \subseteq \safe$, we backtrack and reduce the step size, which improves the precision.
Later, we attempt to increase the step size again to improve scalability.

While this \emph{property-guided} idea is simple, a naive implementation would not scale, especially for high-dimensional systems.
In particular, it is costly to repeatedly compute a discretization whenever the time step changes, among other reasons due to the expensive computation of the matrix exponential.
Our main contribution is an efficient implementation of this idea.
A key insight is to precompute a library of conservative time-discretized systems~\cite{ForetsS22} for different step sizes, which allows us to change the step size fast.
Moreover, for efficient calculation of this library, we use specific step sizes (powers of two; in particular, this choice avoids the repeated computation of the matrix exponential).
Furthermore, to deal with the uncertain inputs~$u(t)$, we propose a hybrid approach that combines small-step and large-step analysis.
Finally, we utilize a mixed set representation when calculating~$\oreach{0}{T}(\X_0)$; specifically, we combine zonotopes and the support function~\cite{Combiningzonotopesandsupportfunctions}, where we use zonotopes for the discretization and the support function for taking the discrete steps.

As the main limitation of our approach, the precision is selected based on the property; hence, the results are not reusable if the property changes.

\paragraph{Contributions.}

First, we present a new algorithm for property-guided reachability analysis of LTI systems~\eqref{eq:LTI}, which significantly reduces the computational verification effort.
Second, we show how to implement conservative time discretization for variable time steps in a scalable way by reusing intermediate results to obtain a more precise sets for larger step sizes while only calculating the matrix exponential once.
Third, we implement our approach efficiently and demonstrate its scalability on a set of high-dimensional benchmark models in comparison to state-of-the-art reachability algorithms.
Our source code and benchmark scripts are publicly available~\cite{ReACT:Github}.

\subsection{Related Work}

\paragraph{Fixed time step.}
The standard way to reachability analysis for LTI dynamics discretizes system~\eqref{eq:LTI} for a small time step~$\delta > 0$~\cite{ForetsS22}.
We shortly sketch the general procedure for the simplified case without inputs~$u(t)$ here (and we recall one concrete way in the next section).
First, one computes an overapproximation~$\oreach{0}{\delta}(\X_0)$, for which a coarse approximation suffices because~$\delta$, and hence the approximation error, is small.
One also computes the matrix~$e^{A \delta}$ (the matrix exponential of~$A \delta$); the motivation is that the linear map~$e^{A \delta} \X_0$ describes the set of reachable states~$\reach{\delta}{\delta}(\X_0)$ at time point~$\delta$.
Finally, one propagates the set~$\Omega_0 = \oreach{0}{\delta}(\X_0)$ forward with the recurrence~$\Omega_{k+1} = e^{A \delta} \Omega_k$ until $k \delta \ge T$.
The sets~$\Omega_k$ are the blocks depicted in \figref{fig:motivation}.

There are many algorithms implementing the above procedure in different ways, for instance varying in the choice of the representation or the way to obtain~$\oreach{0}{\delta}(\X_0)$~\cite{BasicZonotope,BasicNotation,SupportFunctions,SpaceEx:errorsupportFunc,Combiningzonotopesandsupportfunctions,DecomposingReachSetComputations,Luo2023}.
All these approaches have in common that they use a fixed time step~$\delta$.
As a consequence, the user needs to choose~$\delta$ small enough for achieving sufficient precision to prove the safety property; hence, $k$ becomes large in practice, even if most of the sets~$\Omega_k$ are far away from the boundaries of the safe region and would consequently not require high precision.

\paragraph{Adaptive time step.}
Existing reachability algorithms with variable time steps have the goal to achieve a prescribed approximation error~\cite{PrabhakarV11,FrehseKG13,KOCHDUMPER2024101491,ErrorBoundBased,SpaceEx:errorsupportFunc}.
For that, these algorithms estimate the error and adapt the time step accordingly.
Hence, they typically shine if a good approximation is desired throughout the computation.
However, this error-bounded paradigm is decoupled from the verification problem: one first approximates the full reachable states and only checks the safety property afterward.
Thus, if one is interested in solving the verification problem, one has to guess which precision will be needed.
In contrast, we propose a property-driven algorithm, for which the approximation error does not matter as long as it is small enough to prove the property.
We automatically adapt to the error such that the safety property can be proven.
While our algorithm sometimes has to backtrack if the precision is insufficient to prove safety, this backtracking is always by a single  step only, and we only have to approximate the reachable states once in this way, while error-bounded adaptive algorithms would have to redo the whole work if the error parameter was chosen too large.

\paragraph{Time discretization.}
Since the seminal work by Girard~\cite{BasicZonotope}, efficient reachability algorithms use the paradigm of time discretization followed by set propagation, albeit different approximation models have been proposed for the discretization~\cite{ForetsS22}.
Notable approximation models are based on interval matrices~\cite{AlthoffSB07} or an optimization problem over a forward and a backward approximation~\cite{SpaceEx:errorsupportFunc}.
In this paper, we utilize a simplified (forward-only) version of the latter, which avoids the optimization step~\cite{DecomposingReachSetComputations} (see Section~\ref{sec:forward_discretization}).
One of our central contributions is a new approach to compute a time discretization for large time steps with good precision and scalability, which has not been demonstrated before.

\section{Preliminaries}

In this section, we recall the essential steps to derive classical reachability algorithms for LTI systems.
We first explain, assuming some discretization method, how to compute a sequence of sets that contain the reachable states.
Then, we recall one particular such discretization method.
Finally, we discuss two efficient ways to represent and operate with sets: zonotopes and the support function.

\subsection{Approximating the Reachable States}

Recall that we consider the continuous-time LTI dynamics in~\eqref{eq:LTI}, which are characterized by the triple~$(A, \X_0, \U)$, with~$A \in \R^{n \times n}$ and $\X_0, \U \subseteq \R^n$:
\begin{equation}\tag{\ref{eq:LTI}}
    \dot{x}(t) = Ax(t)+u(t), \quad t\in \left[0, T\right]
\end{equation}

Note that this includes systems of the form~$\dot{x}(t) = Ax(t)+Bu(t)+c$ by a standard transformation of the input domain~$\U' \gets B \U \oplus \{c\}$, where~$\oplus$ denotes the Minkowski sum.
The unique solution of ODE system~\eqref{eq:LTI} at time point~$t$ is
\begin{equation}\label{eq:LTIsol}
    x(t)= e^{At}x(0) + \int_0^t e^{A(t-s)}u(s)ds,
\end{equation}
for initial state~$x(0)$ and input signal~$u(t) \in \U$.
We assume that both $\X_0$ and $\U$ are compact and convex sets.
Let $\reach{t_1}{t_2}(\X_0)$ denote the reachable states in a time interval $[t_1, t_2]$ starting from the initial set~$\X_0$; formally:
\begin{equation*}
    \reach{t_1}{t_2}(\X_0) = \{x(\tau) \mid t_1 \le \tau \le t_2, x(0) \in \X_0, u(t) \in \U\}
\end{equation*}
(Our notation omits the implicit dependence of~$\reachSymbol$ on~$\U$ for simplicity.)
The exact set of reachable states for a time interval is often not computable, and one typically overapproximates it with a sequence of sets~$\Omega_0,\dots,\Omega_{N-1}$ called a \emph{reach tube} such that $\reach{t_{k}}{t_{k+1}}(\X_0) \subseteq \Omega_k$ up to time~$T$, where~$t_0 = 0$ and~$t_N = T$.
Thus, from now on we consider the overapproximation~$\oreach{0}{T}(\X_0) = \Omega_0 \cup \dots \cup \Omega_{N-1}$.

Given a set of safe states~$\safe \subseteq \R^n$, the \emph{safety verification problem} is to determine whether~$\reach{0}{T} \subseteq \safe$ holds.
A sufficient condition is~$\oreach{0}{T}(\X_0) \subseteq \safe$, for any overapproximation; however, a failed inclusion check is inconclusive.

As the system is linear, we can utilize the superposition principle of the solution in~\eqref{eq:LTIsol} to compute the \emph{input solution} (derived by letting $x(0)=0$) separately.
For a time step~$\delta > 0$, we have that~\cite[Lemma 1]{SpaceEx:errorsupportFunc}
\begin{equation}\label{eq:separation}
    \reach{t}{t+\delta}(\X_0) = e^{At} \reach{0}{\delta}(\X_0) \oplus \reach{t}{t}(\{0\}).
\end{equation}
Thus, for any time point~$t \ge 0$, to overapproximate the reachable states in the time interval~$[t, t + \delta]$, it suffices to overapproximate the two summands,
which cover the reachable states in the interval~$[0, \delta]$ resp.\ the input solution at time~$t$.
Following notation from~\cite{SpaceEx:errorsupportFunc}, we introduce shorthands for these approximations:
\begin{equation}\label{eq:overapproximation}
\begin{aligned}
    \Omega_{[0, \delta]}(\X_0, \U) &\supseteq \reach{0}{\delta}(\X_0),
    \\[2mm]
    \Psi_{t}(\U) &\supseteq \reach{t}{t}(\{0\}).
\end{aligned}
\end{equation}

Given a time step~$\delta$, we call a triple~$(e^{A \delta}, \Omega_{[0,\delta]}(\X_0, \U), \Psi_\delta(\U))$ satisfying~\eqref{eq:overapproximation} a \emph{discretized system}.
Letting $\Psi_0 = \{0 \}$ and $\delta_k = t_{k+1} - t_k$, Eq.~\eqref{eq:separation} yields a recurrence to compute the reach tube $\Omega_0,\dots,\Omega_{N-1}$~\cite[Eq.~(6)]{SpaceEx:errorsupportFunc}:
\begin{equation}\label{eq:ApproxModel}
  \begin{aligned}
    \Psi_{k+1} &= \Psi_k \oplus e^{At_k}\Psi_{\delta_k}(\U), \\
    \Omega_k &= e^{At_k}\Omega_{[0, \delta_k]}(\X_0, \U) \oplus \Psi_k
  \end{aligned}
\end{equation}

Note that~\eqref{eq:ApproxModel} does not require to use a constant time step~$\delta$.
But when mixing different time steps, one also needs to compute multiple discretized systems.
It can be shown that the reach tube obtained from~\eqref{eq:ApproxModel} for time steps~$\delta_k$ with~$N$ chosen such that~$\sum_{k=0}^{N-1} \delta_k \ge T$ satisfies~\cite[Proposition~1]{SpaceEx:errorsupportFunc}
\begin{equation*}
    \reach{0}{T}(\X_0) \subseteq \bigcup_{k=0}^{N-1} \Omega_k.
\end{equation*}

\subsection{Forward Discretization Method}\label{sec:forward_discretization}

We use the \emph{forward} approximation model~\cite{DecomposingReachSetComputations} (a simplification of the model in~\cite{SpaceEx:errorsupportFunc} which also goes backward in time) for instantiating $\Omega_{[0,\delta]}$ and $\Psi_t$ in~\eqref{eq:overapproximation}.
\begin{align}
    \Psi_\delta (\U) &= \delta\U\oplus E_\Psi(\U, \delta), \label{eq:spaceexdiscreteInput} \\
    \Omega_{[0,\delta]}(\X_0, \U) &= \CHexact(\X_0, e^{A\delta} \X_0 \oplus E^+ (\X_0, \delta) \oplus \Psi_\delta (\U)) \label{eq:spaceexdiscreteFull}
\end{align}
where $\CHexact(\X, \mathcal{Y})$ denotes the convex hull of the union of two sets $\X$ and $\mathcal{Y}$ and
\begin{align*}
    E_\Psi(\U, \delta) &= \boxdot(\Phi_2(\lvert A \rvert, \delta) \; {\boxdot}(A\U)), \\
    E^+(\X_0, \delta) &= \boxdot(\Phi_2(\lvert A \rvert, \delta) \; {\boxdot}(A^2\X_0)),
\end{align*}
where~$\Phi_2(A, \delta) = \sum_{i=0}^{\infty} \frac{\delta^{i+2}}{(i+2)!} A^i$~\cite{SpaceEx:errorsupportFunc}, $\lvert A \rvert$ denotes the component-wise absolute value of~$A$, and~$\boxdot(\X)$ yields the symmetric interval hull of~$\X$, which is the smallest axis-aligned box centered in the origin and containing~$\X$.

\subsection{Zonotopes}\label{sec:zonotopes}

A \emph{zonotope} is the image of a unit hypercube under an affine map.
Equivalently, it is described by a center $c\in \R^n$ and a finite number of generators $g_1,\dots,g_p \in \R^n$:
\begin{equation*}
    \mathcal{Z} = \left\{c+\sum_{i=1}^{p}\alpha_i g_i \mid \alpha_i\in [-1,1]\right\}
\end{equation*}

We use the shorthand notations $\mathcal{Z} = (c, \langle g_1, \dots,g_p\rangle)$ and $\mathcal{Z} = (c, G)$, where $G\in \R^{n\times p}$ is a matrix whose columns are the generators.
The order of a zonotope is $o = \frac{p}{n}$, describing the relative number of generators with respect to the dimensionality of the state space.

Let $\mathcal{Z} = (c, \langle g_1,\dots,g_p\rangle) = (c, G)$ and $\mathcal{Z}' = (c', \langle g_1',\dots,g_\ell'\rangle)$ be two zonotopes.
Their Minkowski sum is the zonotope $\mathcal{Z} \oplus\mathcal{Z}' = (c+c', \langle g_1,\dots,g_p, g_1',\dots,g_\ell'\rangle)$.
Applying a linear map~$M$ to~$\mathcal{Z}$ yields the zonotope $M\mathcal{Z} = \left(Mc, MG \right)$.
The convex hull of (the union of) two zonotopes is generally not a zonotope; we utilize the operator~$\CH(\mathcal{Z}, \mathcal{Z}')$ presented in prior work~\cite{BasicZonotope} that yields a zonotope overapproximating the exact convex hull~$\CHexact(\mathcal{Z}, \mathcal{Z}')$, e.g., in~\eqref{eq:spaceexdiscreteFull}.

The Minkowski sum increases the number of generators, which makes it computationally critical when applied repeatedly.
We use two countermeasures.
First, we apply a classical order-reduction algorithm~\cite{BasicZonotope}, which overapproximates a zonotope by another zonotope with fewer generators~\cite{ReductionMethods}.
Second, we sometimes use the support function (described below) to represent the effect of the Minkowski sum instead of explicitly computing the zonotope.

\subsection{Support Function}

The \emph{support function}~$\rho$ yields the support of a compact convex set~$\X \subseteq \R^n$ in a direction~$\ell \in \R^n$, which is the scalar $\rho_{\X}(\ell) = \max\{\ell^\top x \mid x \in \X \}$. Specifically, the support of a zonotope $\mathcal{Z}=(c,\langle g_1, \ldots, g_p\rangle)$ is $\rho_{\mathcal{Z}}(\ell) = \ell^\top c + \sum_{i =1}^p \vert \ell^\top g_i \vert$~\cite{Combiningzonotopesandsupportfunctions}.
Given compatible compact convex sets $\X, \mathcal{Y}$ and a matrix~$M$, it is easy to evaluate the support function for Minkowski sums and linear maps~\cite[Table~II]{Combiningzonotopesandsupportfunctions}:
\begin{align*}
        \rho_{\X\oplus\mathcal{Y}}(\ell) &=\rho_{\X}(\ell) + \rho_{\mathcal{Y}}(\ell),\\
        \rho_{M\X}(\ell) &= \rho_{\X}(M^\top \ell).
\end{align*}

It has previously been shown that reachability analysis can be performed with the use of the support function~\cite{SupportFunctions}, also in combination with zonotopes~\cite{Combiningzonotopesandsupportfunctions}.
Using the support function has two main benefits: first, evaluating the support function for a small number of directions~$\ell$ is typically cheaper than applying the full set computations. Second, we do not increase the zonotope order when evaluating the support of the Minkowski sum of two zonotopes; instead, we only add the (scalar) support of the sets in the given direction.
We will use zonotopes for the discretization, as they avoid the recursive definition of the support function.

Note that only knowing the support of the sets~$\Omega_k$ in selected directions is not a restriction for us.
We only need to check whether they are contained in the safe region~$\safe$.
Typically, $\safe$ is given as a polyhedron in half-space representation, so it suffices to evaluate~$\Omega_k$ in one direction per half-space.
If $\safe$ is not given in half-space representation, one can approximate~$\Omega_k$ with template polyhedra~\cite{SpaceEx:errorsupportFunc}.

\section{Property-Driven Time-Step Adaptation}

In this section, we describe our adaptive reachability approach with the following roadmap.
We first select a finite set of time steps (Section~\ref{sec:select_delta}).
Next, given an LTI system~\eqref{eq:LTI}, we discretize the continuous dynamics using~\eqref{eq:spaceexdiscreteInput}--\eqref{eq:spaceexdiscreteFull} for each selected time step (Section~\ref{sec:discretize_multiple}).
Then, we present our adaptive algorithm to compute the reach tube using the recurrence in~\eqref{eq:ApproxModel}, where we try to choose the time step as large as possible (Section~\ref{sec:adaptive_reach_tube}).
During this computation, we check whether the set~$\Omega_k$ violates the safety property.
If so, we backtrack and decrease the time step.
This way, we dynamically adjust the time step to verify safety.
Finally, we propose an adaptation policy for when to increase the time step again (Section~\ref{sec:policy}) and discuss how to mitigate high zonotope orders (Section~\ref{sec:order_mitigation}).

\subsection{Selection of Candidate Time Steps}\label{sec:select_delta}

We initially make a selection of potential time steps $\Delta = [\delta_0, \delta_1, \dots, \delta_m]$ (sorted in ascending order).
For each time step, we will need to compute the approximation in~\eqref{eq:spaceexdiscreteFull}.
To avoid excessive computations of the matrix exponential $\phi_{\delta_i} = e^{A\delta_i}$, given a choice for $\delta_0 > 0$, we choose the other step sizes by doubling:
\begin{equation}
    \delta_{i+1} = 2 \delta_i,
    \label{eq:TimeSeries}
\end{equation}
where $(0 \le i < m)$.
Thus, $\Delta$ is defined by the user parameters~$\delta_0$ and~$m$ (technically, $m$ is upper-bounded via the time horizon~$T$).
Let $\bar{\phi}_i = e^{A\delta_i}$ be the matrix exponential for a given $\delta_i$.
Following~\eqref{eq:TimeSeries}, $\bar{\phi}_{i+1} = e^{A\delta_i \cdot 2} = (\bar{\phi}_i)^2$.
Thus, we can compute~$\bar{\phi}_i$ from~$\bar{\phi}_0$ using matrix multiplication instead of computing the usual approximation of the power series, which is particularly expensive in high dimensions.
For convenience, we introduce the aliases~$\delta^- = \delta_0$ and $\delta^+ = \delta_m$.

\begin{remark}
    The recurrence~\eqref{eq:TimeSeries} makes the time steps grow fast (doubling each time).
    Using scaling and squaring, one could easily obtain a more fine-grained sequence that still exploits the same idea to avoid explicitly computing the matrix exponential.
    Still, in our empirical evaluation, we observed that the simpler version in~\eqref{eq:TimeSeries} always sufficed.
    Note that having more candidate step sizes is not always beneficial, as it also increases the choices during backtracking later.
\end{remark}

\subsection{Discretization for Multiple Time Steps}\label{sec:discretize_multiple}

Given an LTI system~\eqref{eq:LTI} of the form~$(A, \X_0, \U)$ and a selection of time steps $\Delta=[\delta_0, \dots,\delta_m]$, we compute the discretized systems $(\phi_\delta, \Omega_{[0,\delta]}(\X_0, \U), \Psi_\delta(\U))$ for all $\delta\in \Delta$ using~\eqref{eq:spaceexdiscreteInput}--\eqref{eq:spaceexdiscreteFull} (recall that~$\phi_\delta = e^{A\delta}$).

\begin{remark}
    While we chose to present our discretization algorithm as a variant of the \emph{forward} method (cf.\ Section~\ref{sec:forward_discretization}), the approach is generic and also applies to other discretization methods~\cite{ForetsS22}.
\end{remark}

\begin{algorithm}[t]
\caption{Naive Discretization}
\label{alg:DiscretizationNaive}
\begin{algorithmic}[1]
\Require $A$, $\X_0$, $\U$, $\delta^-$, $\delta^+$
\State $\delta \leftarrow \delta^-$
\label{Alg:DiscNaivePsiCalc}
\While{$\delta \leq \delta^+$} 
\State $\phi[\delta] \leftarrow e^{A\delta}$  \Comment{Compute matrix exponential for each $\delta \in \Delta$}
\State $\Psi \gets \text{Compute using \eqref{eq:spaceexdiscreteInput}}$ 
\State $\Omega\leftarrow \text{Compute using \eqref{eq:spaceexdiscreteFull}}$
\State $X[\delta] \leftarrow \Omega$
\State $V[\delta] \leftarrow \Psi$
\State $\delta \leftarrow 2\delta$ \label{Alg:DiscNaiveStep}
\EndWhile
\State\Return $(\phi, X, V)$
\end{algorithmic}
\end{algorithm}

Algorithm~\ref{alg:DiscretizationNaive} shows a naive baseline to compute the discretized systems.
The discrete systems are represented as three maps~$\phi$, $X$, and~$V$, which for each value~$\delta \in \Delta$ store the corresponding matrix exponential~$\phi_{\delta}$ respectively the sets~$\Omega_{[0, \delta]}$ and~$\Psi_\delta$.
Observe that the algorithm is agnostic to the fact that the time step always doubles in line~\ref{Alg:DiscNaiveStep}.
In particular, the matrix exponential (and, hidden in the references to~\eqref{eq:spaceexdiscreteInput}--\eqref{eq:spaceexdiscreteFull}, the matrix~$\Phi_2$) is computed for each value~$\delta$ from scratch.
Moreover, since the approximation error of the discrete systems grows as the time steps grow larger, the results for larger~$\delta$ are useless in practice.

\begin{algorithm}[h!]
\caption{Optimized Discretization}
\label{alg:Discretization}
\begin{algorithmic}[1]
\Require $A$, $\X_0$, $\U$, $\delta^-$, $\delta^+$
\State $\delta \leftarrow \delta^-$
\State $\phi[\delta] \leftarrow e^{A\delta}$ \label{alg:phicomp}
\State $\Psi \leftarrow \text{Compute using \eqref{eq:spaceexdiscreteInput}}$ \label{alg1:PsiCompOnce}
\State $\Omega\leftarrow \text{Compute using \eqref{eq:spaceexdiscreteFull}}$ \label{alg1:OmegaCompOnce}
\While{$\delta \leq \delta^+$} \label{alg1:while}
\State $X[\delta] \leftarrow \Omega$ \label{alg:Xsave}
\State $V[\delta] \leftarrow \Psi $\label{alg:Vsave}
\State $\Omega \leftarrow \CH(\Omega, \phi[\delta] \Omega \oplus \Psi)$ \Comment{Improved two-step recurrence~\eqref{eq:ApproxModel}} \label{alg1:XCalc}
\State $\Psi \leftarrow \Psi \oplus \phi[\delta]\Psi$ \label{alg1:PsiComp}
\State $\delta \leftarrow 2\delta $
\State $\phi[\delta] \leftarrow (\phi[\frac{\delta}{2}])^2$ \label{alg:phicalc}
\EndWhile
\State\Return $(\phi, X, V)$
\end{algorithmic}
\end{algorithm}

Algorithm~\ref{alg:Discretization} shows our optimized discretization method, which exploits the knowledge that the time steps always double to improve both efficiency and precision.
The algorithm utilizes two crucial insights.
First, as mentioned, despite using multiple time steps~$\delta_i$, we only need to compute the matrix exponential~$\phi_{\delta_0}$ \emph{once} from scratch (line~\ref{alg:phicomp}) and then can utilize matrix multiplication (line~\ref{alg:phicalc}).
The second insight is that it suffices to use the approximation model~\eqref{eq:spaceexdiscreteInput}--\eqref{eq:spaceexdiscreteFull} for the smallest time step~$\delta^-$ (lines~\ref{alg1:PsiCompOnce}--\ref{alg1:OmegaCompOnce}), for which the approximation error is typically small.
For larger time steps, we instead propagate the sets for the previous time step one step forward using the recurrence~\eqref{eq:ApproxModel} (lines~\ref{alg1:XCalc}--\ref{alg1:PsiComp}).

\begin{figure}[t]
    \begin{subfigure}{0.49\textwidth}
        \includegraphics[width=\linewidth]{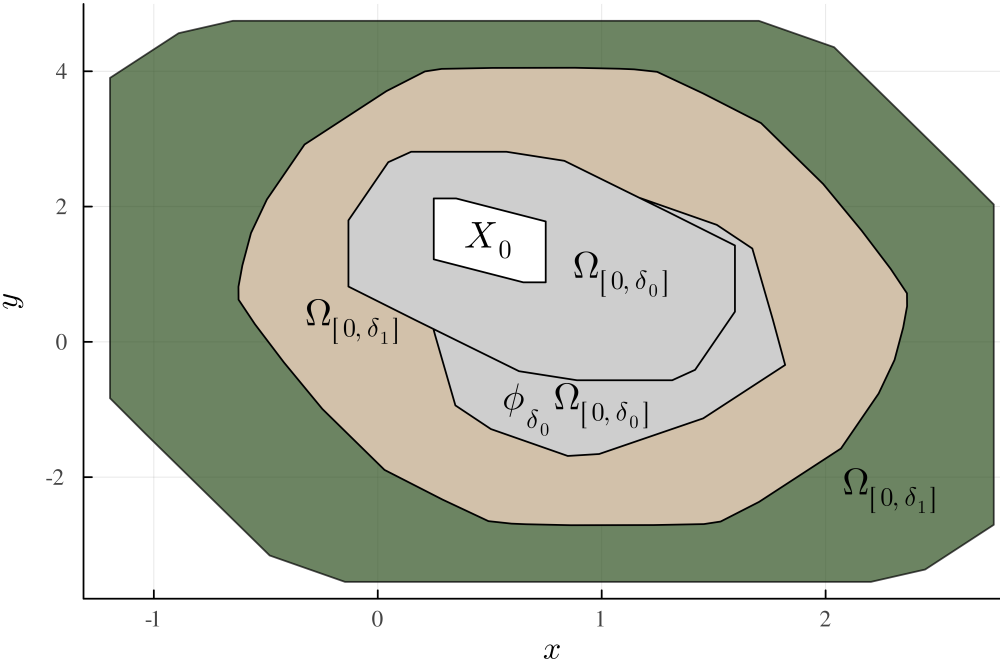}
        \caption{Without zonotope approximation.}
        \label{subfig:discLazy}
    \end{subfigure}
    \hfill
    \begin{subfigure}{0.49\textwidth}
        \includegraphics[width=\linewidth]{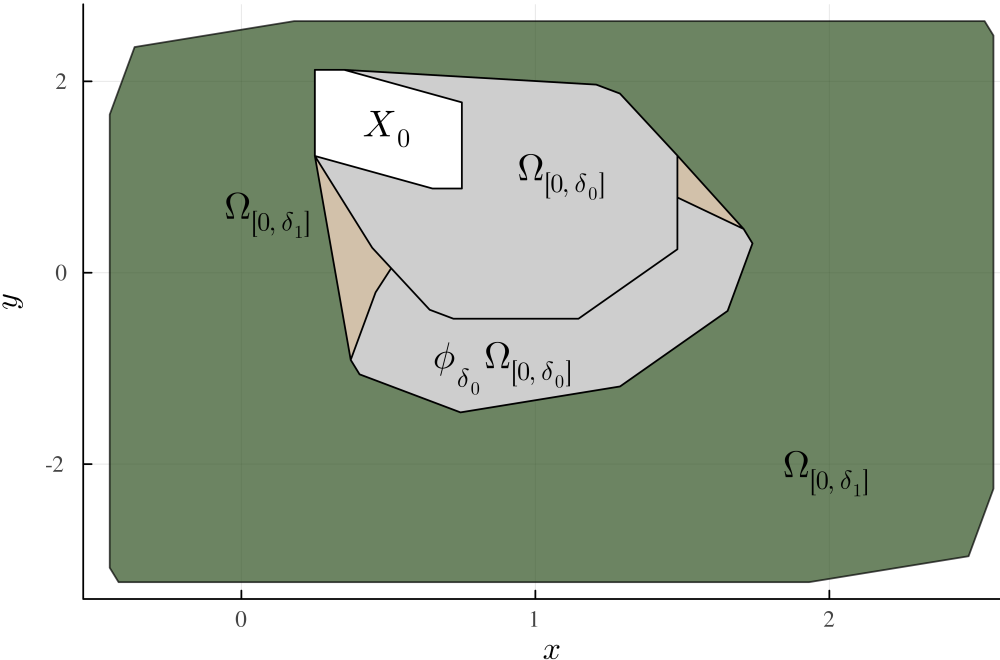}
        \caption{With zonotope approximation.}
        \label{subfig:disc}
    \end{subfigure}
    \caption{Comparison of the discretization algorithms~\ref{alg:DiscretizationNaive} and~\ref{alg:Discretization} for a system with similar dynamics as in \figref{fig:motivation} with initial states~$\X_0$ (white) and~$\U = \{0\}$.
    On the left, we use the zonotope approximation ($\CH$).
    On the right, we use the exact convex-hull operator ($\CHexact$).
    Algorithm~\ref{alg:DiscretizationNaive} computes $\Omega_{[0, \delta_1]}$ from scratch (green).
    Algorithm~\ref{alg:Discretization} computes $\Omega_{[0, \delta_1]}$ as the convex hull of $\Omega_{[0, \delta_0]} \cup \phi_{\delta_0} \Omega_{[0, \delta_0]}$ (brown).}
    \label{fig:discretization}
\end{figure}

We illustrate the difference between Algorithm~\ref{alg:DiscretizationNaive} and Algorithm~\ref{alg:Discretization} in \figref{fig:discretization} for the case~$\U = \{0\}$.
In theory, the outputs are incomparable because the better approximation of the iterative Algorithm~\ref{alg:Discretization} may be partially lost due to the additional approximation when taking the zonotope approximation of the convex hull in line~\ref{alg1:XCalc} (cf.\ Section~\ref{sec:zonotopes}).
In practice, Algorithm~\ref{alg:Discretization} yields a much smaller approximation error, in particular for the larger time steps.
This is because the common approximation models (like the one described in Section~\ref{sec:forward_discretization} that we use here) are typically only precise for small time steps~\cite{ForetsS22}.

\subsection{Reach Tube Approximation with Adaptive Time Step}\label{sec:adaptive_reach_tube}

After obtaining the discrete systems, we use Algorithm~\ref{alg:ReachTube} to compute the approximation of the reach tube.
Additional inputs are the safe region~$\safe$, which we assume to be polyhedral, along with the list of directions~$L$ that describe its facets, the time horizon~$T$ and the available time steps~$\Delta$, and a time-step policy~$\policy$ which we explain later.

\begin{algorithm}[t!]
\caption{Computation of the Reach Tube}\label{alg:ReachTube}
\begin{algorithmic}[1]
\Require ($\phi$, $X$, $V$): discretized systems,
\Require $\safe, L$: safe region and collection of corresponding direction vectors,
\Require $T, \Delta$: time horizon and available time steps,
\Require $\policy$: policy for increasing time step
\For{$\ell \in L$}
\State $\Psi_0[\ell] \gets 0$
\EndFor
\State $t \gets 0$, $k\gets 0$, $\delta \gets \delta^+$, $\Phi \gets I$
\While{$t < T$}\label{alg:timehorizon} \Comment{Invariant: $\Phi = e^{A t}$}
\While{$\delta \geq \delta^-$} \label{alg:mainwhile}
\For{$\ell \in L$} \Comment{Evaluate $\rho_{\Psi_{t+\delta} (\U)}(\ell)$ and $\rho_{\Omega_{[t, t+\delta]} (\X_0 ,\U)}(\ell)$}
\State $\Psi_{k+1}[\ell] \gets \Psi_{k}[\ell] + \rho_{V[\delta]}(\Phi^\top \ell)$ \label{alg:psiminkowski}
\State $\Omega_{k}[\ell] \gets \rho_{X[\delta]}(\Phi^\top \ell) + \Psi_{k}[\ell]$ \label{alg:omega}
\EndFor
\If{$ \Omega_{k}[\ell] \le \rho_{\safe}(\ell) \text{ for all } \ell \in L$} \Comment{Check $\Omega_k\subseteq\safe$}\label{alg:constraintcheck}
\State $t \gets t+\delta$
\State $k \gets k+1$
\State $\Phi \gets \Phi \phi[\delta]$ \Comment{$e^{A (t+\delta)}=e^{At}e^{A \delta}$}\label{alg:Phi}
\State $\textbf{break}$
\Else \Comment{Decrease time step and try again}
\State $\delta \gets \frac{1}{2}\delta$ \label{alg:lowerstep}
\EndIf
\EndWhile
\If{$\delta < \delta^-$} \Comment{$\Omega_{k} \not\subseteq \safe$ at maximal precision}
\State\Return ``Could not verify safety'' \label{alg:terminateUnknown}
\Else \Comment{May try a different step size next time}
\State $\delta \gets \text{Adjust according to $\policy$}$ \label{algline:adjusttimestep}
\EndIf
\EndWhile
\State\Return ``System is safe'' \label{alg:terminateSafe}
\end{algorithmic}
\end{algorithm}

Algorithm~\ref{alg:ReachTube} is inspired by the well-known reachability algorithm based on the support function~\cite{SupportFunctions} (however, that algorithm uses a fixed time step).
The motivation to use the support function is to avoid (1)~the explicit computation of the reach tube~$\Omega_0, \dots, \Omega_{N-1}$ and (2)~the emergence of high-order zonotopes, which would have to be tamed by order reduction (and thus additional approximation errors).
Instead, we only evaluate the support function for the reach tube in selected directions, which is sufficient to check inclusion in the safe region~$\safe$.

The loop in line~\ref{alg:mainwhile} computes the next set~$\Omega_k$ of the reach tube by starting with the current time step~$\delta$.
For each attempted~$\delta$, we compute bounds for the candidate set~$\Omega_k$ (line~\ref{alg:omega}) and perform the safety check (line~\ref{alg:constraintcheck}).
If~$\Omega_k$ is safe, we accept the time step~$\delta$ and update~$\Phi$, which stores the cumulative matrix exponential, for the next iteration (line~\ref{alg:Phi}).
Otherwise, if $\Omega_k$ is not safe, we try the next-smaller time step~$\delta$ (line~\ref{alg:lowerstep}).
This process repeats until either~$t\geq T$ (successful verification) or~$\delta < \delta^-$ (failed verification) holds.

Clearly, Algorithm~\ref{alg:ReachTube} terminates because there are finitely many choices for~$\delta$.
The following theorem asserts that the computed result is correct; moreover, the algorithm can prove safety for the same systems as the induced fixed-step approach that uses the same discretization for the step size~$\delta = \delta^-$ all the time.
Note that this holds even if Algorithm~\ref{alg:ReachTube} has not used~$\delta^-$ all the time.

\begin{theorem}[Soundness and relative completeness]\label{thm:sound_complete}
    If Algorithm~\ref{alg:ReachTube} returns in line~\ref{alg:terminateSafe}, the system is safe.
    If Algorithm~\ref{alg:ReachTube} returns in line~\ref{alg:terminateUnknown}, the standard fixed-step algorithm using~$\delta^-$ at all times would also fail.
\end{theorem}

\begin{proof}[Sketch]
    The first property holds trivially since we check safety for each step.
    The second property holds because (1)~we always backtrack when a safety check fails and (2)~right before Algorithm~\ref{alg:ReachTube} fails, we use the smallest step~$\delta^-$ and~$t$ is a multiple of~$\delta^-$.
    The latter holds because our time steps are multiples of~$\delta^-$.
    Hence, the last set we compute equals precisely the set the fixed-step algorithm would compute at this point, and thus that algorithm would also fail.
\end{proof}
\subsection{Adaptation Policy}\label{sec:policy}

While Algorithm~\ref{alg:ReachTube} yields sound results under any policy~$\policy$ to adjust the next time step~$\delta_i$ to a value in~$\Delta$, the policy impacts the efficiency.
In our implementation, we use a policy that dynamically adjusts~$\delta_i$ based on previous adjustments.
First, our policy never decreases~$\delta_i$.
Moreover, if the time step did not have to be reduced in the previous four steps, we double the time step (provided this would not exceed~$\delta^+$).
We found that this policy works well in practice.

\subsection{Mitigating the Growth of the Zonotope Order}\label{sec:order_mitigation}

As mentioned, when using zonotopes, the increase in generators caused by the repeated Minkowski additions can become a bottleneck.
We use the support function to represent~$\Psi_k$ in Algorithm~\ref{alg:ReachTube} (line~\ref{alg:psiminkowski}).
Furthermore, by accumulating~$\Phi$ in line~\ref{alg:Phi}, we avoid additional order increases in line~\ref{alg:omega}, since $\Omega_k$ is not recursively defined.
However, order reductions can still be necessary in Algorithm~\ref{alg:Discretization} for larger time steps, as that algorithm uses the recurrence formulation.

\begin{remark}
    The repeated convex hull computations can also significantly increase the zonotope order when~$\delta^+ \gg \delta^-$.
    As evaluating the support function of a convex hull can be done exactly, this error could also be mitigated.
    However, empirically we found that this negatively impacts the efficiency of the algorithm.
\end{remark}

\section{Evaluation}
In this section, we evaluate our implementation of the proposed approach, which we refer to as \emph{\TheTitle} in the following (but which also includes Algorithm~\ref{alg:Discretization}), in the Julia programming language.
For all results, we use a machine with an Intel i5 11400F CPU and 16 GB RAM running Windows and using Julia v1.12.

\subsection{Benchmark Problems and Baseline Algorithms}

\begin{table}[t]
    \caption{SLICOT benchmark statistics.
    The specifications are taken from~\cite{TranNJ16,BakD17}.
    Some specifications use a linear combination of the state variables, denoted by~$y$.}
    \label{tab:BenchmarkModelOverview}
    \centering
    \begin{tabular}{c c c}
        \toprule
        Model & Dimension ($n$) & Safety property ($\safe$) \\
        \midrule
        Motor & \phantom{00}8 &  $x_1\in \left[0.35, 0.4\right] \vee x_5\notin \left[0.45, 0.6\right]$ \\
        Building & \phantom{0}48 & $x_{25} < 6\cdot10^{-3}$ \\
        PDE & \phantom{0}84 & $y < 12$\\
        Heat & 200 & $x_{133} < 0.1$\\
        ISS & 270 & $y \in \left[-7, 7\right] \cdot 10^{-4}$\\
        Beam & 384 & $x_{89} < 2100$\\
        MNA1 & 578 & $x_1 < 0.5$\\
        \bottomrule
    \end{tabular}
\end{table}

We consider a subset of the SLICOT benchmark models~\cite{bencmarkmodels}, which represent high-dimensional real-world applications.
All models are LTI systems of the form~\eqref{eq:LTI} and are summarized in Table~\ref{tab:BenchmarkModelOverview}.
The time horizon is always~$T = 20$.

We compare to two state-of-the-art reachability algorithms that use a fixed time step.
The first algorithm is the one based on the support function~\cite{SupportFunctions,LGG09}, which we denote as \emph{LGG}.
The second algorithm is based on decomposition~\cite{DecomposingReachSetComputations}, which we denote as \emph{BFFPSV}.
For both LGG and BFFPSV, we apply the same algorithm options as in a prior evaluation~\cite{DecomposingReachSetComputations} and use the \emph{forward} discretization.
We employ the implementations of these algorithms in the tool JuliaReach~\cite{JuliaReach}, which is also written in Julia, for a fair comparison; for instance, we use the same method for calculating the matrix exponential for all three approaches.

\subsection{Reach Tube Computation and Verification of Safety Properties}

We report about the results of our evaluation in Table~\ref{tab:mainResults}.
Since LGG and BFFPSV first compute the full reach tube and then verify safety (whereas in our approach, these steps are interleaved), we report the total time for these tools.

We used the same fixed time steps~$\delta$ for LGG and BFFPSV as in the prior evaluation~\cite[Table~4]{DecomposingReachSetComputations}.
We selected the step ranges ($\delta^-$ and $\delta^+$) such that~$\delta^- = \delta$; by Theorem~\ref{thm:sound_complete}, we thus can verify the same safety properties as the baselines.

We highlight that we often choose large values for~$\delta^+$ that would normally seem unrealistic (e.g., $\delta^+ = 1.024$ for the Heat model, which is close to~$T = 20$).
This highlights the precision we can achieve thanks to Algorithm~\ref{alg:Discretization}.

\begin{table*}[t]
    \caption{Evaluation on the SLICOT benchmarks. Bold entries mark the best result and $\dagger$ marks a failure to verify the property. \TheTitle uses $\delta^- = \delta$, where $\delta$ is the fixed time step, and $\delta^+$ is given as a power of two larger than $\delta$.}
    \label{tab:mainResults}
    \centering
    \resizebox{\textwidth}{!}{\begin{tabular}{c c c r c c c c c r}
        \toprule
        & & \multicolumn{3}{c}{\TheTitle (ours)} & & LGG & & BFFPSV & \\
        \cline{3-5} \cline{7-10} 
        Model & $\delta^- = \delta$ & $\delta^+/\delta^-$ & Steps & Time (s) & & Time (s) & &  Time (s) &  Steps\\
        \midrule
        Motor & $1\cdot10^{-3}$ & $2^3$  & 2503 & {\boldmath$1.82\cdot 10^{-2}$} &  & $6.45\cdot 10^{-2}$ &  & $1.94\cdot 10^{-1}$ & $20000$\\ 

        Building & $2\cdot10^{-3}$ & $2^{9}$ & 239 & {\boldmath $9.35\cdot 10^{-3}$} & & $2.49\cdot 10^{-1}$ & & $1.46\cdot 10^{-1}$ & 10000\\ 

        PDE & $3\cdot10^{-4}$ & $2^{10}$  & 81 & {\boldmath $4.01\cdot 10^{-2}$}  & & $5.67$ & & $322$ & 66667\\ 

        Heat & $1\cdot10^{-3}$ & $2^{10}$ & 28 & {\boldmath $1.83\cdot 10^{-1}$}  & & $9.00$ & & $4.71$ & 20000\\ 

        ISS & $6\cdot10^{-4}$ & $ 2^5$ & 1042 & {\boldmath $3.25\cdot 10^{-1}$}  & & $2.87$ & & $\dagger~(88.5)$ & 33334\\

        Beam & $5\cdot10^{-5}$ & $2^5$ & 12501  & {\boldmath $4.44$} & & $778$ & & $343$ & 400000\\ 

        MNA1 & $4\cdot10^{-4}$ & $2^{11}$ & 971 &  {\boldmath$2.57$}  & & $688$ &  & $123$ & 50000\\ 
        \bottomrule
    \end{tabular}}
\end{table*}

Our approach consistently outperforms the baselines across all benchmarks, often achieving dramatic speedups of one or two orders of magnitude.
The main reason is that we can effectively employ large time steps most of the time, whereas the fixed-step approaches must use tiny steps all the time.
For instance, we are able to verify the 578-dimensional MNA1 model in less than~$3$ seconds, which is a $\times 47$ speedup over BFFPSV and a $\times 267$ speedup over LGG.

For the ISS model, BFFPSV is not precise enough to prove the property.
This is consistent with previous results~\cite{DecomposingReachSetComputations}.
LGG is able to verify the property, and does so efficiently due to the sparsity of the system matrix.
However, our approach still outperforms LGG with a~$\times 8$ speedup.

\subsection{Approximation Quality}

\begin{figure}[t!]
    \begin{subfigure}{0.48\textwidth}
        \includegraphics[width=\linewidth]{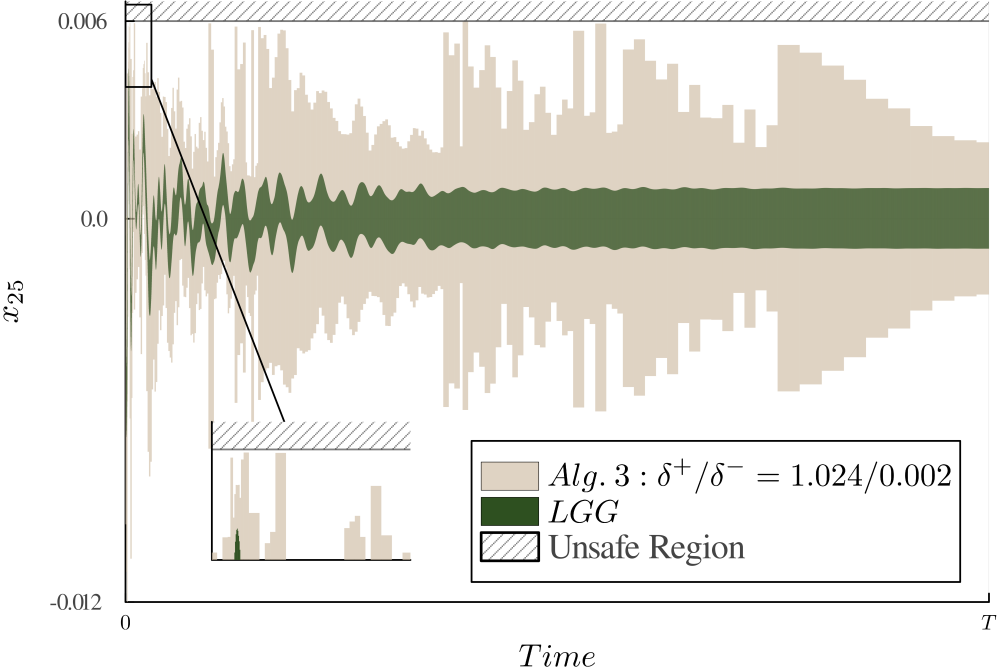}
        \caption{Building model.}
        \label{fig:Building}
    \end{subfigure}
    \hfill
    \begin{subfigure}{0.48\textwidth}
        \includegraphics[width=\linewidth]{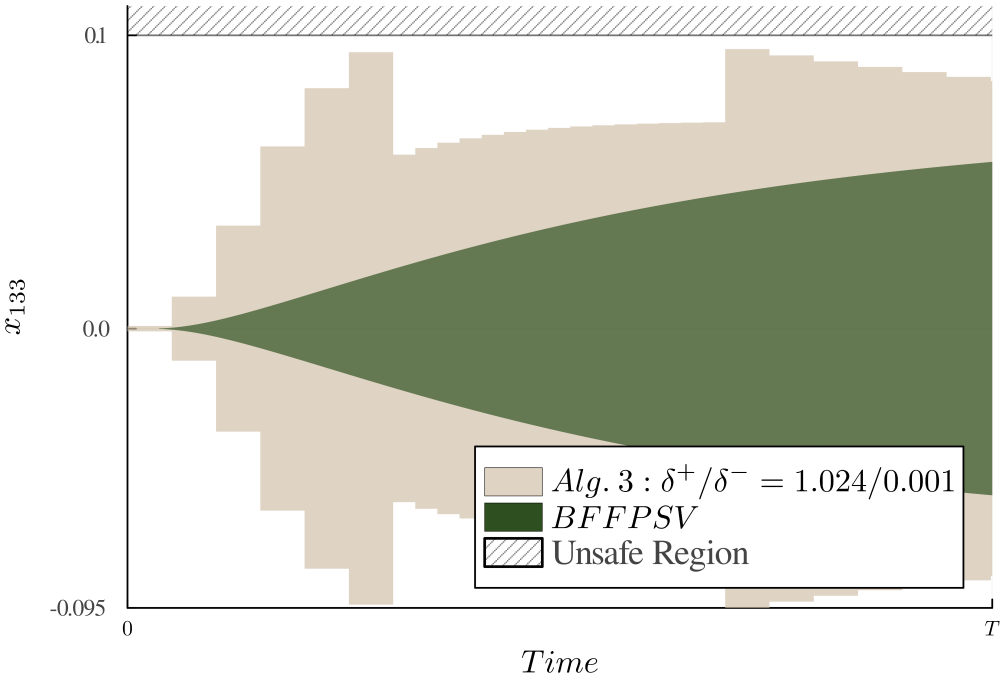}
        \caption{Heat model.}
        \label{fig:Heat}
    \end{subfigure}
    \caption{Comparison of \TheTitle to LGG and BFFPSV. The step sizes for all algorithms are identical with those reported in Table~\ref{tab:mainResults}.}
    \label{fig:ReACT-LGG-BFFPSV-comparison}
\end{figure}

In contrast to the baselines, our approach is willing to produce large approximation errors as long as the specification can be proven.
Hence, a direct comparison of the approximation quality is not meaningful.
Nevertheless, we show how the approximations differ empirically for three selected benchmark models.
(Since the approximation quality of LGG and BFFPSV is typically similar to each other, we only show one of these approaches per plot.)

For the Building model (\figref{fig:Building}), the reach tube computed by LGG is closest to the unsafe region in the first few
steps, after which the high precision would not be needed anymore.
\TheTitle gets close to the unsafe region (and hence switches the time step) many times.
The reach tube computed by \TheTitle only consists of 239 steps, whereas LGG computes for 10{,}000 steps.

The Heat model (\figref{fig:Heat}) highlights multiple benefits of \TheTitle. We can take large time steps because of the high precision of our discretization. \TheTitle verifies the Heat model in 28 steps, while BFFPSV computes for 20{,}000 steps. Additionally, the adaptive approach allows \TheTitle to continuously use the best time step from the selection~$\Delta$.

\begin{figure}[ht]
    \centering
    \includegraphics[width=\textwidth,height=70mm,keepaspectratio]{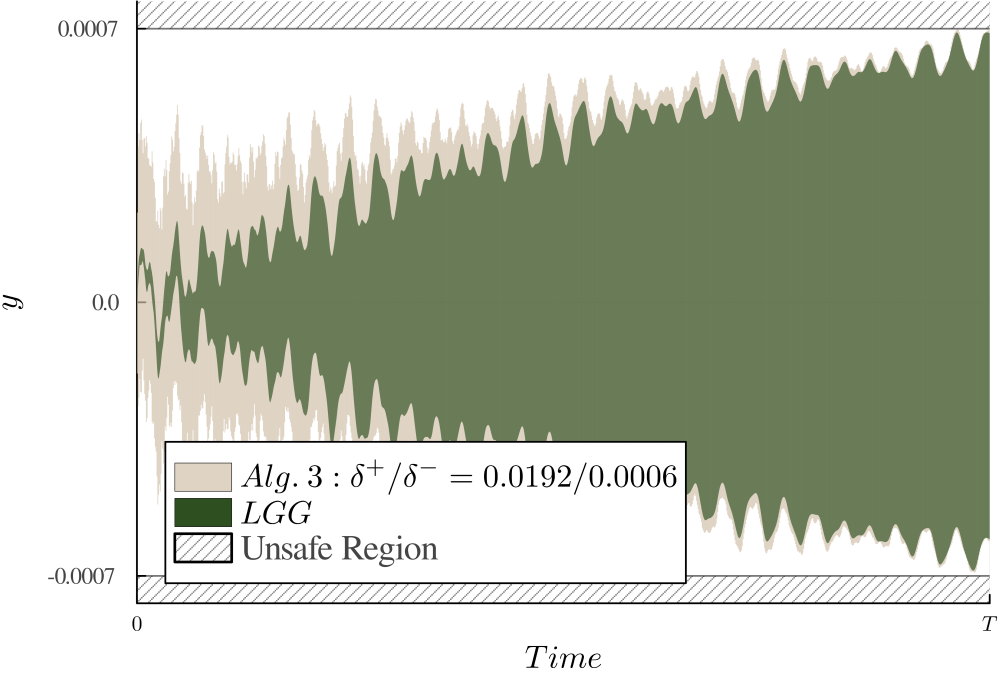}
    \caption{Comparison of \TheTitle to LGG on the ISS model.
    The step sizes for the algorithms are identical with those reported in Table~\ref{tab:mainResults}.}
    \label{fig:ISSbenchmark}
\end{figure}

For the ISS model (\figref{fig:ISSbenchmark}), the green sets computed by LGG get close to the unsafe region.
This explains why BFFPSV, which is generally less precise than LGG, cannot prove this specification.
(Note that using smaller time steps would not necessarily help, as that requires more time steps and hence more error accumulation.)
\TheTitle can generally take larger time steps in the beginning.
The reach tube computed by \TheTitle only consists of 1{,}042 steps, whereas LGG computes for 33{,}334 steps.

\subsection{Property-Driven Adaptation}

\figref{fig:Building} demonstrates how our approach benefits from property-driven time-step adaptation on the Building model.
We see that the time step evolves over time for a model where precision is crucial early in the time interval~$[0,T]$.
As time advances, the distance between the green reach tube and the unsafe region becomes larger.
This allows our approach to gradually increase~$\delta$, which drastically decreases computation time.
Notably, the largest time step we use is~$\times 512$ larger than the smallest one that is necessary to prove safety.

\subsection{Comparison of Algorithm~\ref{alg:DiscretizationNaive} and Algorithm~\ref{alg:Discretization}}

\begin{figure}[t]
    \centering
    \includegraphics[width=\textwidth,height=70mm,keepaspectratio]{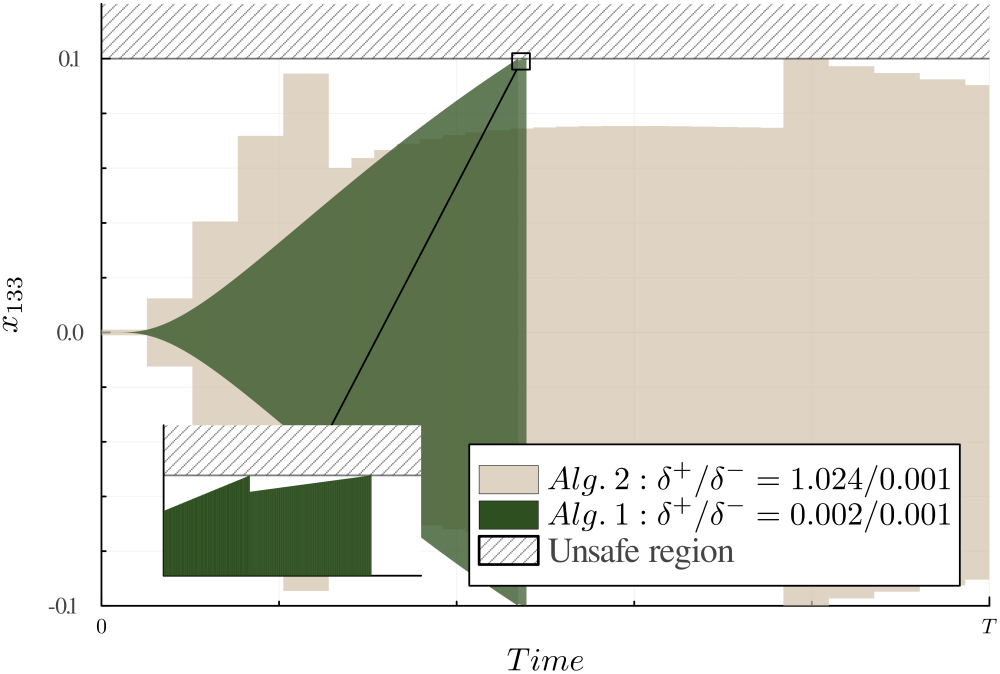}
    \caption{Comparison of the discretization algorithms~\ref{alg:DiscretizationNaive} and~\ref{alg:Discretization} on the Heat model. Algorithm~\ref{alg:Discretization} uses $\delta^-=0.001$ and $\delta^+=1.024$, while Algorithm~\ref{alg:DiscretizationNaive} uses $\delta^-=0.001$ and $\delta^+=0.002$.}
    \label{fig:discretizationComparison}
\end{figure}

As mentioned, our optimized discretization algorithm empirically achieves lower approximation errors compared to the naive algorithm, even at large time steps.
This allows us to sometimes even use the largest time step~$\delta^+$ for several steps during the reach tube computation in \TheTitle.
In \figref{fig:discretizationComparison}, we demonstrate the difference of using the optimized discretization (Algorithm~\ref{alg:Discretization}) over the naive discretization (Algorithm~\ref{alg:DiscretizationNaive}).
Both algorithms use the smallest time step~$\delta^- = 0.001$.
With Algorithm~\ref{alg:Discretization} and~$\delta^+ = 1.024$, our method successfully verifies the model, for which it only uses the largest time steps~$0.512$ and~$1.024$. Even when limiting Algorithm~\ref{alg:DiscretizationNaive} to~$\delta^+ = 0.002$, our method fails within 10 time units because the precision is not high enough.
This is due to the accumulation of the approximation errors in~$\Psi$ when using Algorithm~\ref{alg:DiscretizationNaive}.

\section{Conclusion}

We have proposed an efficient reachability algorithm for continuous-time linear time-invariant systems with uncertain inputs.
Our algorithm is based on time discretization and uses an adaptive step size.
In contrast to other adaptive algorithms, ours is property-driven instead of error-driven: the precision is increased only when necessary to prove the specification.
To achieve efficient performance, we also developed a novel discretization method that allows us to make large time steps in practice.
We have implemented our algorithms efficiently and demonstrated performance improvements over state-of-the-art methods of one or two orders of magnitude on high-dimensional benchmark systems.

In the future, it will be interesting to investigate how our approach can be used for reach-avoid specifications.
We also plan to generalize it to hybrid systems.
There we can partially reuse our algorithm and treat transition guards as unsafe regions.
In contrast, however, once an intersection with the guard is detected, we are in a new setting.
Thus, a key challenge will be to avoid that our discretization algorithm becomes a bottleneck after each discrete transition, both in terms of performance and overapproximation.

\subsection*{Acknowledgments}
This research was partly supported by the Independent Research Fund Denmark under reference number 10.46540/3120-00041B and the Villum Investigator Grant S4OS under reference number 37819.

\bibliographystyle{splncs04}
\bibliography{main}

\end{document}